\documentclass{article}
\usepackage[hidelinks]{hyperref}
\usepackage[utf8]{inputenc}
\usepackage{authblk}
\usepackage{amsthm}
\usepackage{amsmath}
\usepackage{amssymb}
\usepackage{tikz}
\usepackage{algorithm}
\usepackage{algpseudocode}
\usetikzlibrary{graphs}
\theoremstyle{plain}

\pgfdeclarelayer{background}
\pgfdeclarelayer{foreground}
\pgfsetlayers{background,main,foreground}
\tikzset{
    my box/.style = {
        , line cap = round
        , line join = round
    }
  }
\newcommand{\highlight}[3]{
  \path [my box, line width = #1, draw = #2, transparency group, opacity=1] #3;
}
\usetikzlibrary{arrows.meta}
\def\niceArrow{-{Stealth[length=2.25mm]}}
\usepackage{subcaption}

\usepackage[margin=1cm]{caption}

\newtheorem{proposition}{Proposition}
\newtheorem{corollary}{Corollary}
\newtheorem{claim}{Claim}
\newtheorem{lemma}{Lemma}
\newtheorem{theorem}{Theorem}

\theoremstyle{definition}

\newtheorem{example}{Example}

\newtheorem{note}{Note}

\newcommand{\FBOX}{\hspace*{\fill}$\rule{0.17cm}{0.17cm}$}

\usepackage[a4paper, total={6.9in, 9.75in}]{geometry}

\usepackage{enumitem}
\setlist[itemize]{leftmargin=*}

\algnewcommand{\True}{\textbf{true}}
\algnewcommand{\False}{\textbf{false}}
\algnewcommand{\Continue}{\textbf{continue}}
\algnewcommand{\Break}{\textbf{break}}
\algnewcommand{\Not}{\textbf{not}\ }
\algnewcommand{\And}{\textbf{and}\ }
\newcommand{\Input}{\item[\textbf{Input:}]}
\newcommand{\Output}{\item[\textbf{Output:}]}
\newcommand{\Effect}{\item[\textbf{Effect:}]}


\date{}

\usepackage[yyyymmdd]{datetime}

\usepackage{titling}
\thanksmarkseries{arabic}

\title{Quadratic-Time Algorithm for the\\ Maximum-Weight $(k, \ell)$-Sparse Subgraph Problem}

\author{
  Bence De\'ak\thanks{Department of Operations Research, E{\"o}tv{\"o}s Lor{\'a}nd University, Budapest, Hungary. E-mail: \texttt{deakbence2002@gmail.com}}
  \and P\'eter Madarasi\thanks{HUN-REN Alfr\'{e}d R\'{e}nyi Institute of Mathematics, and Department of Operations Research, E{\"o}tv{\"o}s Lor{\'a}nd University, Budapest, Hungary. E-mail: \texttt{madarasip@staff.elte.hu}}
}

\begin{document}

\maketitle

\vspace{-5mm}
\begin{abstract}
  The family of $(k, \ell)$-sparse graphs, introduced by Lorea, plays a central role in combinatorial optimization and has a wide range of applications, particularly in rigidity theory.
  A key algorithmic challenge is to compute a maximum-weight $(k, \ell)$-sparse subgraph of a given edge-weighted graph.
  Although prior approaches have long provided an $O(nm)$-time solution, a previously proposed $O(n^2 + m)$ method was based on an incorrect analysis, leaving open whether this bound is achievable.

  We answer this question affirmatively by presenting the first $O(n^2 + m)$-time algorithm for computing a maximum-weight $(k, \ell)$-sparse subgraph, which combines an efficient data structure with a refined analysis.
  This quadratic-time algorithm enables faster solutions to key problems in rigidity theory, including computing minimum-weight redundantly rigid and globally rigid subgraphs.
  Further applications include enumerating non-crossing minimally rigid frameworks and recognizing kinematic joints.
  Our implementation of the proposed algorithm is publicly available online.

\end{abstract}
\bigskip

\section{Introduction}\label{sec:intro}
Throughout this paper, let $k$ and $\ell$ be non-negative integers with $\ell < 2k$, let $G = (V, E)$ be a loopless multigraph with $n = |V|$ vertices and $m = |E|$ edges, and assume that the edges $e_1, \dots, e_m$ are sorted in non-increasing order with respect to a given weight function.
The graph $G = (V, E)$ is called \emph{$(k, \ell)$-sparse} if, for each subset $X \subseteq V$, the number $i(X)$ of edges induced by $X$ is at most $\max\{k|X| - \ell, 0\}$.
Furthermore, if $G$ is $(k, \ell)$-sparse and has exactly $\max\{k |V| - \ell, 0\}$ edges, then we say that $G$ is \emph{$(k, \ell)$-tight}.
A graph is \emph{$(k, \ell)$-spanning} if it contains a $(k, \ell)$-tight subgraph that spans the entire vertex set.
A \emph{$(k, \ell)$-block} of a $(k, \ell)$-sparse graph is a subset $X \subseteq V$ that induces a $(k, \ell)$-tight subgraph.
A \emph{$(k, \ell)$-component} is an inclusion-wise maximal $(k, \ell)$-block.
These definitions form the foundation of a rich theory, closely connected to matroids and rigidity, and they naturally lead to fundamental optimization problems that are the focus of this paper.

\smallskip
The concept of $(k, \ell)$-sparse graphs was first introduced by Lorea~\cite{lorea} as part of his work on matroidal families.
Since then, these graphs have been the subject of extensive research, with numerous applications in various areas of mathematics and computer science.
For example, $(k, k)$-tight graphs appeared in the work of Nash-Williams~\cite{nash1961edgeDisjointST} and Tutte~\cite{tutte} as a characterization of graphs that can be decomposed into $k$ edge-disjoint spanning trees.
A classical result in rigidity theory, due to Laman~\cite{laman}, characterizes generic minimally rigid bar-joint frameworks in the plane as $(2, 3)$-tight graphs, while the rigid graphs correspond to the $(2, 3)$-spanning graphs.
For a comprehensive overview of rigidity theory, see~\cite{jordan2016combinatorial,schulze2017rigidity}.

A wide range of optimization and decision problems can be reduced to the \emph{maximum-size $(k, \ell)$-sparse subgraph problem}, where $k$ and $\ell$ are fixed non-negative integers determined by the specific task.
The goal is to find a $(k, \ell)$-sparse subgraph of a given graph $G = (V, E)$ that contains the maximum number of edges.
A natural generalization is the \emph{maximum-weight $(k, \ell)$-sparse subgraph problem}, where the goal is to find a $(k, \ell)$-sparse subgraph with maximum total weight under a given weight function on the edges.
Many classical problems in graph theory can be viewed as special cases of this framework --- for example, finding a maximum-weight spanning tree, a maximum-weight subgraph that can be decomposed into $k$ forests, or a maximum-weight rigid spanning subgraph in the plane.

Both the maximum-size and the maximum-weight $(k, \ell)$-sparse subgraph problems can be solved using the pebble game algorithms introduced in~\cite{bergPhD, berg2003algorithms, hendrickson, pebble, pebbleDS}.
The formulation in~\cite{berg2003algorithms} provides a more convenient orientation-based view of these algorithms, which we adopt throughout this paper.
The naive implementation of these algorithms runs in $O(nm)$ time, where $n$ and $m$ denote the number of vertices and edges, respectively.
A widely accepted approach in the literature claimed to improve this to $O(n^2 + m)$ by employing a more sophisticated data structure~\cite{pebble,pebbleDS}, building on ideas in~\cite{berg2003algorithms}.
In recent years, however, several researchers have pointed out that the running time analysis of the improved method is flawed~\cite{madarasi2023klSparse, mihalyko2022augmentation}, and the question of whether such a time bound is truly achievable has remained open.

\smallskip
In our paper, we provide a positive answer to this question: we present an algorithm that solves the maximum-weight $(k, \ell)$-sparse subgraph problem in $O(n^2 + m)$ time.
Our approach combines a carefully designed data structure with a refined analysis, leading to both theoretical and practical advances.
As a direct corollary, the maximum-weight rigid subgraph problem --- corresponding to the well-known case $k = 2$, $\ell = 3$ --- can now be solved in quadratic time.
Beyond its standalone significance, the maximum-weight $(k, \ell)$-sparse subgraph problem also frequently arises as a subroutine in more complex combinatorial optimization problems.
For example, our algorithm substantially improves the running time of approximation algorithms for the minimum-weight redundantly rigid and globally rigid subgraph problems~\cite{jordan2020minimum} and their generalizations~\cite{mihalyko2022augmentation} in the metric case.
Further applications include enumerating non-crossing minimally rigid frameworks~\cite{avis2008enumerating}, as well as recognizing kinematic joints~\cite{john2012kinematic}.
An efficient implementation of the algorithms discussed in this paper is publicly available online~\cite{githubSparse}, and a detailed empirical evaluation can be found in~\cite{madarasi2025efficientKL}.

\smallskip
The rest of the paper is organized as follows.
Section~\ref{sec:naive_algo} provides an overview of the classical augmenting-path algorithm --- often referred to as the ``naive'' pebble game --- for computing maximum-size and maximum-weight $(k, \ell)$-sparse subgraphs.
In Section~\ref{sec:comp_algo}, we present our component-based variant of this algorithm, which forms the core of our contributions.
This high-level approach relies on two auxiliary procedures for managing $(k, \ell)$-components, whose implementation can be tailored to the specific setting of the problem.
Our main contribution is presented in Section~\ref{sec:alg_general}, where we give a complete implementation of these subroutines for the general range $0 \leq \ell < 2k$.
Combining structural properties of $(k, \ell)$-components with an efficient data structure, we prove that the resulting algorithm achieves a running time of $O(n^2 + m)$, improving on the $O(nm)$ bound of the naive approach.
Finally, Section~\ref{sec:spec_cases} addresses two notable special cases.
First, we examine the case $0 \leq \ell \leq k$, where the disjointness of $(k, \ell)$-components allows for a simplified implementation reducing the space complexity from $O(n^2)$ to $O(n)$.
Second, we address the unweighted variant for the general range, in which a vertex-wise edge-processing strategy retains the $O(n^2 + m)$ running time while requiring only $O(n)$ space.
These results together establish a unified algorithmic framework for efficiently solving the maximum-size and maximum-weight $(k, \ell)$-sparse subgraph problem across the full range of parameters.

\subsection{The naive augmenting-path algorithm}\label{sec:naive_algo}
We overview the classical, or ``naive'', version of the pebble game algorithm for finding a maximum-size or maximum-weight $(k, \ell)$-sparse subgraph.
We begin by stating two fundamental results that are essential for proving the correctness of the algorithm, and then present the algorithm itself.

\begin{theorem}[Lorea, 1979~\cite{lorea}]
\label{theorem:matroid}
Let $G = (V, E)$ be a graph, and define $\mathcal{I} = \{F \subseteq E : (V, F) \text{ is } (k, \ell)\text{-sparse} \}$.
Then $M = (E, \mathcal{I})$ is a matroid.
\FBOX
\end{theorem}

We also need the following technical result, which can be derived from Hakimi's Orientation Lemma~\cite{SLHOrientationLemma}.
Here we give a direct proof.
Note that a similar lemma appeared as Lemma~1 in~\cite{berg2003algorithms} for $(k, \ell)=(2, 3)$.
\begin{lemma}\label{lemma:sparse_orientation}
Let $H = (V, F)$ be a $(k, \ell)$-sparse graph, and let $u,v \in V$ be two distinct vertices.
Suppose $D$ is an orientation of $H$ in which each vertex has indegree at most $k$.
Then:
\begin{enumerate}
    \item If $\varrho_D(u) + \varrho_D(v) < 2k - \ell$, then $H' = (V, F \cup \{uv\})$ is also $(k, \ell)$-sparse.
    \item If $\varrho_D(u) + \varrho_D(v) \geq 2k - \ell$, and there is no directed path in $D$ from any vertex $w \in V \setminus \{u, v\}$ with $\varrho_D(w) < k$ to either $u$ or $v$, then $H' = (V, F \cup \{uv\})$ is not $(k, \ell)$-sparse.
\end{enumerate}
\end{lemma}
\begin{proof}
  We prove the two parts separately.

  \smallskip
  1.\ Since $H$ is $(k, \ell)$-sparse, it suffices to verify the sparsity condition for subsets $X \subseteq V$ containing both $u$ and $v$, as these are the only sets affected by the insertion of the edge $uv$.
  For any such $X$,
  \[
    i_{H'}(X) = i_H(X) + 1 \leq \sum_{w \in X} \varrho_D(w) + 1 \leq k(|X| - 2) + (2k - \ell - 1) + 1 = k|X| - \ell,
  \]
  which confirms that $H'$ remains $(k, \ell)$-sparse.

  \smallskip
  2.\ Let $T \subseteq V$ denote the set of vertices from which $u$ or $v$ is reachable in $D$.
  Clearly, $\varrho_D(T) = 0$, and by assumption, $\varrho_D(w) = k$ for each $w \in T \setminus \{u,v\}$.
  Hence,
  \[
    i_{H'}(T) = i_H(T) + 1 = \sum_{w \in T} \varrho_D(w) - \varrho_D(T) + 1 \geq (2k - \ell) + k(|T| - 2) + 1 = k|T| - \ell + 1,
  \]
  which violates the sparsity condition, proving that $H'$ is not $(k, \ell)$-sparse.
\end{proof}

\paragraph{Checking edge acceptability via orientations.}
Given a $k$-indegree-bounded orientation $D$ of $H$, Lemma~\ref{lemma:sparse_orientation} yields an algorithm for checking whether the insertion of an edge $uv$ preserves $(k, \ell)$-sparsity.
While $\varrho_D(u) + \varrho_D(v) \geq 2k - \ell$, we find a directed path $P$ in $D$ from a vertex in $\{w \in V \setminus \{u, v\} : \varrho_D(w) < k\}$ to either $u$ or $v$.
If no such path exists, then the insertion of $uv$ violates the sparsity condition.
Otherwise, we reverse the arcs of $P$, which decreases $\varrho_D(u) + \varrho_D(v)$ without violating the $k$-indegree bound.
Repeating this process at most $\ell + 1$ times determines whether the edge $uv$ can be inserted.
$\bullet$
\medskip

We are now ready to describe the augmenting-path algorithm for finding a maximum-weight $(k, \ell)$-sparse subgraph.
The algorithm is built on two key ideas:
\begin{itemize}
\item We iteratively construct the optimal subgraph $H = (V, F)$, considering the edges one by one in the order $e_1, \dots, e_m$.
  Each edge $e$ is inserted into $H$ if and only if the resulting subgraph $H' = (V, F \cup \{e\})$ remains $(k, \ell)$-sparse.
\item We maintain a $k$-indegree-bounded orientation $D$ of the current subgraph $F$.
  For each edge $e$, we apply Lemma~\ref{lemma:sparse_orientation} together with the orientation-based procedure described above to determine whether the edge can be accepted.
  If $H' = (V, F \cup \{e\})$ is $(k, \ell)$-sparse, then we insert $e$ into $H$ and update $D$ by inserting $e$ with an orientation that preserves the $k$-indegree bound.
\end{itemize}

The correctness of this method immediately follows by Theorem~\ref{theorem:matroid} and Lemma~\ref{lemma:sparse_orientation}.
Algorithm~\ref{alg:basic} provides an implementation of this iterative approach.
\begin{algorithm}[H]
  \caption{Naive Pebble Game for $0 \leq \ell < 2k$}\label{alg:basic}
  \begin{algorithmic}[1]
    \Input An undirected graph $G = (V, E)$.
    \Require The edges $e_1, \dots, e_m$ of $G$ are sorted in non-increasing order of weight.
    \Output The edge set of a maximum-weight $(k, \ell)$-sparse subgraph.
    \vspace{0.2em}
    \hrule
    \vspace{0.2em}
    \Procedure{PebbleGame$_{k,\ell}$}{$G = (V, E)$}
      \State $F \gets \emptyset$ \Comment{Initialize the edge set of the optimal subgraph}
      \State $D \gets (V, \emptyset)$ \Comment{Initialize a directed graph on $V$ without arcs}
      \For{$e \gets e_1, \dots, e_m$} \Comment{Process the edges in non-increasing order of weight}\label{alg:basic:for}
        \State $u, v \gets \operatorname{endpoints}(e)$
      	\While{$\varrho_D(u) + \varrho_D(v) \geq 2k - \ell$}\label{alg:basic:while}
      	  \State find a path $P$ in $D$ from $\{w \in V \setminus \{u, v\} : \varrho_D(w) < k\}$ to $\{u, v\}$\label{alg:basic:findPath}
      	  \If{no such path exists}
      	    \State \Break \Comment{Exit the \texttt{while} loop}\label{alg:basic:reject}
      	  \EndIf
      	  \State reverse the arcs of $P$ in $D$ \Comment{Decrease the indegree of $u$ or $v$}\label{alg:basic:reverse}
      	\EndWhile
      	\If{$\varrho_D(u) + \varrho_D(v) < 2k - \ell$}\label{alg:basic:acceptIf}
          \State $F \gets F \cup \{e\}$ \Comment{Accept edge $e$}
      	  \If{$\varrho_D(v) < k$}
      	    \State $D \gets D \cup \{uv\}$ \Comment{Insert an arc from $u$ to $v$ in $D$}
      	  \Else
      	    \State $D \gets D \cup \{vu\}$ \Comment{Insert an arc from $v$ to $u$ in $D$}
      	  \EndIf
      	\EndIf\label{alg:basic:acceptEndIf}
      \EndFor
      \State \Return $F$ \Comment{Return the edge set of the optimal subgraph}
    \EndProcedure
  \end{algorithmic}
\end{algorithm}
\begin{proposition}
Algorithm~\ref{alg:basic} runs in $O(nm)$ time.
\end{proposition}
\begin{proof}
The \texttt{while} loop in line~\ref{alg:basic:while} takes at most $\ell + 1$ iterations per edge, and each augmenting path $P$ can be found via a single traversal of $D$ in $O(n)$ time.
Thus, each edge is processed in $O(n)$ time overall, yielding a total running time of $O(nm)$.
\end{proof}

\section{The component-based algorithm}\label{sec:comp_algo}
In Algorithm~\ref{alg:basic}, checking for each edge whether it can be inserted into $D$ without violating the sparsity condition takes $O(nm)$ time in total.
If we instead had access to an oracle that could determine in constant time whether an edge can be accepted, then augmenting paths would only need to be computed for the $O(n)$ edges actually inserted, reducing the total running time to $O(n^2 + m)$.
We construct such an oracle by leveraging the following structural properties of $(k, \ell)$-blocks and $(k, \ell)$-components.

\begin{lemma}
\label{lemma:intersection}
Let $X$ and $Y$ be $(k, \ell)$-blocks of a $(k, \ell)$-sparse graph $H = (V, F)$ such that $|X \cap Y| \geq 2$.
Then both $X \cap Y$ and $X \cup Y$ are $(k, \ell)$-blocks.
\end{lemma}
\begin{proof}
Let $X, Y \subseteq V$ be distinct blocks of $H$ such that $|X \cap Y| \geq 2$.
Recall that $i_H(X)$ denotes the number of edges induced by $X$.
By the supermodularity of this function, we have
\[
  k(|X| + |Y|) - 2 \ell = i_H(X) + i_H(Y) \leq i_H(X \cap Y) + i_H(X \cup Y) \leq k|X \cap Y| + k|X \cup Y| - 2 \ell = k(|X| + |Y|) - 2 \ell.
\]
Since the left- and right-hand sides are equal, all inequalities must hold with equality, and hence $X \cap Y$ and $X \cup Y$ are blocks.
\end{proof}

\begin{corollary}
\label{corollary:comp_intersection}
Any two distinct $(k, \ell)$-components in a $(k, \ell)$-sparse graph $H = (V, F)$ intersect in at most one vertex.
\end{corollary}
\begin{proof}
Let $X$ and $Y$ be components with $|X \cap Y| \geq 2$.
By the previous lemma, $X \cup Y$ is also a block.
Since $X \cup Y$ contains both $X$ and $Y$, it follows that $X = Y$, otherwise $X$ or $Y$ is not inclusion-wise maximal.
\end{proof}

\begin{corollary}
\label{corollary:comp_arises_unique}
Inserting a single edge into a $(k, \ell)$-sparse graph $H = (V, F)$ can form at most one new $(k, \ell)$-component.
\end{corollary}
\begin{proof}
Let $uv$ be the edge to be inserted, and let $X$ and $Y$ be two newly formed components.
Since $u, v \in X \cap Y$, we clearly have $|X \cap Y| \geq 2$.
By Corollary~\ref{corollary:comp_intersection}, this implies that $X = Y$.
\end{proof}

\begin{corollary}
\label{corollary:comp_tot_size}
The total size of the $(k, \ell)$-components in a $(k, \ell)$-sparse graph $H = (V, F)$ is $O(n)$.
\end{corollary}
\begin{proof}
By Corollary~\ref{corollary:comp_intersection}, each edge in $F$ is induced by at most one component.
Since $|F| \leq \max\{ kn - \ell, 0 \} = O(n)$, the number of components is $O(n)$.
A component inducing $t$ edges has at most $(t + \ell)/k < t/k + 2$ vertices, so the sum of the sizes of all components is at most $|F|/k + 2 \cdot O(n) = O(n)$.
\end{proof}

Let us now revisit our original problem and modify Algorithm~\ref{alg:basic} to achieve a faster running time.
Our goal is to decide in constant time whether an edge $uv$ can be inserted into the current subgraph $H = (V, F)$ without violating the sparsity condition.
Clearly, $H' = (V, F \cup \{uv\})$ is $(k, \ell)$-sparse if and only if no $(k, \ell)$-component in $H$ contains both $u$ and $v$.
To determine whether $uv$ can be accepted, the algorithm must therefore check for the existence of such a $(k, \ell)$-component.
If $uv$ is inserted, then a new $(k, \ell)$-component containing both $u$ and $v$ may form, and some existing $(k, \ell)$-components may cease being $(k, \ell)$-components, as they remain $(k, \ell)$-blocks but are no longer maximal.
We introduce three subroutines to handle these operations efficiently:
\begin{itemize}
  \item \textsc{InCommonComponent}$(u, v)$ --- checks whether there is a $(k, \ell)$-component containing both $u$ and $v$.
  \item \textsc{FindComponent}$(D, u, v)$ --- returns the new $(k, \ell)$-component formed after inserting $uv$, or $\emptyset$ if no such $(k, \ell)$-component appears.
  \item \textsc{UpdateComponents}$(C)$ --- updates the data structures representing the current set of $(k, \ell)$-components after a new $(k, \ell)$-component $C$ is identified.
\end{itemize}
Temporarily treating these as black boxes, we obtain the following high-level algorithm.

\begin{algorithm}[H]
\caption{Component Pebble Game for $0 \leq \ell < 2k$}
\label{alg:comp}
  \begin{algorithmic}[1]
    \Input An undirected graph $G = (V, E)$.
    \Require The edges $e_1, \dots, e_m$ of $G$ are sorted in non-increasing order of weight.
    \Output The edge set of a maximum-weight $(k, \ell)$-sparse subgraph.
    \vspace{0.2em}
    \hrule
    \vspace{0.2em}
    \Procedure{ComponentPebbleGame$_{k,\ell}$}{$G = (V, E)$}
    \State $F \gets \emptyset$ \Comment{Initialize the edge set of the optimal subgraph}
    \State $D \gets (V, \emptyset)$ \Comment{Initialize a directed graph on $V$ without arcs}
    \For{$e \gets e_1, \dots, e_m$} \Comment{Process the edges in non-increasing order of weight}\label{alg:comp:for}
        \State $u, v \gets \operatorname{endpoints}(e)$
        \If{\Not \Call{InCommonComponent$_{k,\ell}$}{$u, v$}} \Comment{No component contains both $u$ and $v$?}\label{alg:comp:check_in_comp}
          \While{$\varrho_D(u) + \varrho_D(v) \geq 2k - \ell$}
            \State find a path $P$ in $D$ from $\{w \in V \setminus \{u, v\}$ : $\varrho_D(w) < k\}$ to $\{u, v\}$\label{alg:comp:path_find}
            \State reverse the arcs of $P$ in $D$ \Comment{Decrease the indegree of $u$ or $v$}\label{alg:comp:reverse}
          \EndWhile
          \State $F \gets F \cup \{e\}$ \Comment{Accept edge $e$}
          \If{$\varrho_D(v) < k$}
            \State $D \gets D \cup \{uv\}$ \Comment{Insert an arc from $u$ to $v$ in $D$}
          \Else
            \State $D \gets D \cup \{vu\}$ \Comment{Insert an arc from $v$ to $u$ in $D$}
          \EndIf
          \State $C \gets$ \Call{FindComponent$_{k,\ell}$}{$D, u, v$} \Comment{Find the new component if there is one}
          \If{$C \neq \emptyset$} \Comment{New component found?}
            \State \Call{UpdateComponents$_{k,\ell}$}{$C$} \Comment{Update component data}
          \EndIf
        \EndIf
    \EndFor
    \State \Return $F$ \Comment{Return the set of accepted edges}
    \EndProcedure
  \end{algorithmic}
\end{algorithm}

\begin{note} In line \ref{alg:comp:path_find}, we can always find such a path $P$.
  This follows from Lemma~\ref{lemma:sparse_orientation} and the fact that the edge $uv$ can be inserted into $H$.
\end{note}

The subroutines \textsc{InCommonComponent} and \textsc{UpdateComponents} will be described in full detail in the next section.
In Section~\ref{sec:spec_cases}, we present more efficient implementations tailored to two special cases.
Before presenting the implementation of \textsc{FindComponent} --- which is shared across all special cases --- we first analyze how the $(k, \ell)$-blocks of a $(k, \ell)$-sparse graph change when a new edge is inserted.
This analysis relies on the following two lemmas.
\begin{lemma}
\label{lemma:opt_cond}
Let $D$ be a $k$-indegree-bounded orientation of a $(k, \ell)$-sparse graph $H = (V, F)$ and let $u, v \in V$ be distinct vertices such that $\varrho_D(u) + \varrho_D(v) \leq 2k - \ell$.
A set $X \subseteq V$ containing both $u$ and $v$ is a $(k, \ell)$-block if and only if $\varrho_D(u) + \varrho_D(v) = 2k - \ell$, $\varrho_D(X) = 0$, and each vertex in $X \setminus \{u, v\}$ has indegree $k$.
\end{lemma}
\begin{proof}
  A set $X$ containing $u$ and $v$ is a block if and only if $i_H(X) = k|X| - \ell$.
  By the definition of $i_H(X)$, we have
  $$
    i_H(X) = \sum_{w \in X} \varrho_D(w) - \varrho_D(X) \leq \sum_{w \in X} \varrho_D(w) \leq k(|X| - 2) + 2k - \ell = k|X| - \ell.
  $$
Equality holds if and only if $\varrho_D(X) = 0$, $\varrho_D(u) + \varrho_D(v) = 2k - \ell$, and each vertex in $X \setminus \{u, v\}$ has indegree~$k$.
\end{proof}

Now we derive a necessary and sufficient condition for the existence of a component containing two chosen vertices $u$ and $v$.
A similar lemma was proven as Lemma~2 in~\cite{berg2003algorithms}.
\begin{lemma}
  \label{lemma:new_comp_arises}
  Let $D$ be a $k$-indegree-bounded orientation of a $(k, \ell)$-sparse graph $H = (V, F)$ and let $u, v \in V$ be distinct vertices such that $\varrho_D(u) + \varrho_D(v) = 2k - \ell$.
  Define $S = \{w \in V \setminus \{u, v\} : \varrho_D(w) < k\}$, and let $T$ be the set of vertices not reachable from $S$ in $D$.
  There exists a $(k, \ell)$-component containing $u$ and $v$ if and only if both $u$ and $v$ lie in $T$.
  If such a component exists, then it is exactly $T$.
\end{lemma}
\begin{proof}
Suppose that $u, v \in T$.
Since $S$ and $T$ are disjoint, we have $\varrho_D(w) = k$ for each $w \in T \setminus \{u,v\}$.
Furthermore, $T$ has no incoming edges, so $\varrho_D(T)=0$.
By Lemma~\ref{lemma:opt_cond}, $T$ is a $(k, \ell)$-block, and hence $u$ and $v$ lie in a common component.

Conversely, suppose that there is a component $X$ containing both $u$ and $v$.
By Lemma~\ref{lemma:opt_cond}, we have $\varrho_D(X)=0$, and $\varrho_D(w)=k$ for each $w \in X \setminus \{u,v\}$.
Thus, the vertices in $X$ are unreachable from $S$, implying that $u,v \in X \subseteq T$.
Note that $T$ itself is a $(k, \ell)$-block, as it satisfies the conditions of Lemma~\ref{lemma:opt_cond}.
From $X \subseteq T$ and the maximality of $X$, it follows that $X = T$.
\end{proof}



Using Lemmas \ref{lemma:opt_cond} and \ref{lemma:new_comp_arises}, the \textsc{FindComponent} subroutine can be implemented as follows.

\begin{algorithm}[H]
\caption{Find Component for $0 \leq \ell < 2k$}
\label{alg:find_comp}
  \begin{algorithmic}[1]
    \Input A directed graph $D = (V, A)$, and two distinct vertices $u, v \in V$.
    \Require $D$ is a $k$-indegree-bounded orientation of a $(k, \ell)$-sparse graph $H$, and $\varrho_D(u) + \varrho_D(v) \leq 2k - \ell$.
    \Output The $(k, \ell)$-component in $H$ containing both $u$ and $v$ if it exists; $\emptyset$ otherwise.
    \vspace{0.2em}
    \hrule
    \vspace{0.2em}
    \Procedure{FindComponent$_{k,\ell}$}{$D = (V, A), u, v$}
    \If{$\varrho_D(u) + \varrho_D(v) < 2k - \ell$} \Comment{Cannot form new component}
      \State \Return $\emptyset$
    \EndIf
    \State $Z \gets $ set of vertices reachable from $\{w \in V \setminus \{u, v\} : \varrho_D(w) < k\}$ in $D$
    \State $T \gets V \setminus Z$
    \If{$u, v \in T$} \Comment{New component formed?}
      \State \Return $T$
    \Else
      \State \Return $\emptyset$
    \EndIf
    \EndProcedure
  \end{algorithmic}
\end{algorithm}

\paragraph{Implementation details for \textnormal{\textsc{FindComponent}}.}
The subroutine above can be implemented using a single graph traversal in $O(n)$ time.
By Lemmas \ref{lemma:opt_cond} and \ref{lemma:new_comp_arises}, the procedure returns the unique new $(k, \ell)$-component if one is formed, and $\emptyset$ otherwise.
$\bullet$
\medskip

\begin{example}
  Consider the graph in Figure~\ref{fig:compMerge} with $(k, \ell) = (2, 3)$, and suppose we want to insert the edge $uv$.
  First, we find a $2$-indegree-bounded orientation $D$ with $\varrho_D(u) + \varrho_D(v) < 1$ and add the arc $uv$ to $D$.
  Next, we run a graph traversal from the set of vertices with indegree lower than $2$.
  Let $Z$ denote the set of reached vertices, and define $T = V \setminus Z$.
  Since $u,v\in T$, the newly formed component is exactly $T$.
  Finally, we delete any $(k, \ell)$-components fully contained in $T$, as they are no longer inclusion-wise maximal.

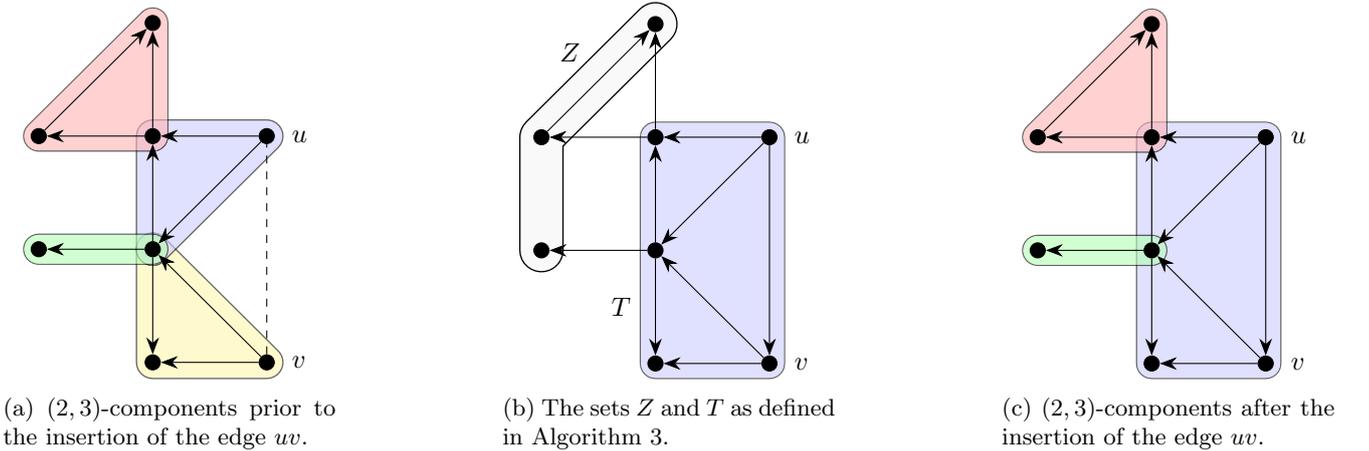
\begin{figure}[H]
  \begin{subfigure}[t]{.25\textwidth}
    \centering
    \begin{tikzpicture}[scale=1.5]
      \node[draw, circle, fill=black, inner sep=2pt] (v0) at (0, 0) {};
      \node[draw, circle, fill=black, inner sep=2pt] (v1) at (0, 1) {};
      \node[draw, circle, fill=black, inner sep=2pt, label={[xshift=1mm]right:$v$}] (v2) at (1, 0) {};
      \node[draw, circle, fill=black, inner sep=2pt] (v3) at (0, 2) {};
      \node[draw, circle, fill=black, inner sep=2pt, label={[xshift=1mm]right:$u$}] (v4) at (1, 2) {};
      \node[draw, circle, fill=black, inner sep=2pt] (v5) at (-1, 1) {};
      \node[draw, circle, fill=black, inner sep=2pt] (v6) at (0, 3) {};
      \node[draw, circle, fill=black, inner sep=2pt] (v7) at (-1, 2) {};

      \draw[\niceArrow] (v2) to (v0);
      \draw[\niceArrow] (v2) to (v1);
      \draw[\niceArrow] (v4) to (v1);
      \draw[\niceArrow] (v4) to (v3);
      \draw[\niceArrow] (v1) to (v0);
      \draw[\niceArrow] (v1) to (v3);
      \draw[\niceArrow] (v1) to (v5);
      \draw[\niceArrow] (v3) to (v7);
      \draw[\niceArrow] (v3) to (v6);
      \draw[\niceArrow] (v7) to (v6);
      \draw[dashed] (v2) to (v4);

      \begin{pgfonlayer}{background}
        \begin{scope}[opacity=.6, transparency group]
          \path[line cap=round, line join=round, line width=0.5pt, double=yellow!40, double distance=4.1mm, draw=black]
            (v1.center) to (v2.center) to (v0.center) to (v1.center);
          \highlight{3.8mm}{yellow!40, fill=yellow!40}{(v1.center) to (v2.center) to (v0.center)}
        \end{scope}
        \begin{scope}[opacity=.6, transparency group]
          \path[line cap=round, line join=round, line width=0.5pt, double=blue!20, double distance=4.1mm, draw=black]
            (v1.center) to (v3.center) to (v4.center) to (v1.center);
          \highlight{3.8mm}{blue!20, fill=blue!20}{(v1.center) to (v3.center) to (v4.center)}
        \end{scope}
        \begin{scope}[opacity=.6, transparency group]
          \path[line cap=round, line join=round, line width=0.5pt, double=red!30, double distance=3.8mm, draw=black]
            (v3.center) to (v6.center) to (v7.center) to (v3.center);
          \highlight{3.8mm}{red!30, fill=red!30}{(v3.center) to (v6.center) to (v7.center)}
        \end{scope}
        \begin{scope}[opacity=.6, transparency group]
          \path[line cap=round, line join=round, line width=0.5pt, double=green!30, double distance=3.8mm, draw=black]
            (v1.center) to (v5.center);
          \highlight{3.8mm}{green!30, fill=green!30}{(v1.center) to (v5.center)}
        \end{scope}
      \end{pgfonlayer}
    \end{tikzpicture}
    \caption{$(2, 3)$-components prior to the insertion of the edge $uv$.}
  \end{subfigure}
  \hfill
  \begin{subfigure}[t]{.25\textwidth}
    \centering
    \begin{tikzpicture}[scale=1.5]
      \node[draw, circle, fill=black, inner sep=2pt] (v0) at (0, 0) {};
      \node[draw, circle, fill=black, inner sep=2pt] (v1) at (0, 1) {};
      \node[draw, circle, fill=black, inner sep=2pt, label={[xshift=1mm]right:$v$}] (v2) at (1, 0) {};
      \node[draw, circle, fill=black, inner sep=2pt] (v3) at (0, 2) {};
      \node[draw, circle, fill=black, inner sep=2pt, label={[xshift=1mm]right:$u$}] (v4) at (1, 2) {};
      \node[draw, circle, fill=black, inner sep=2pt] (v5) at (-1, 1) {};
      \node[draw, circle, fill=black, inner sep=2pt] (v6) at (0, 3) {};
      \node[draw, circle, fill=black, inner sep=2pt] (v7) at (-1, 2) {};

      \draw[\niceArrow] (v2) to (v0);
      \draw[\niceArrow] (v2) to (v1);
      \draw[\niceArrow] (v4) to (v1);
      \draw[\niceArrow] (v4) to (v3);
      \draw[\niceArrow] (v1) to (v0);
      \draw[\niceArrow] (v1) to (v3);
      \draw[\niceArrow] (v1) to (v5);
      \draw[\niceArrow] (v3) to (v7);
      \draw[\niceArrow] (v3) to (v6);
      \draw[\niceArrow] (v7) to (v6);
      \draw[\niceArrow] (v4) to (v2);

      \node at (-.75, 2.75){$Z$};
      \node at (-.3, .5){$T$};

      \begin{pgfonlayer}{background}
        \begin{scope}[opacity=.6, transparency group]
          \path[line cap=round, line join=round, line width=0.5pt, double=blue!20, double distance=3.8mm, draw=black]
            (v1.center) to (v3.center) to (v4.center) to (v2.center) to (v0.center) to (v1.center);
          \highlight{3.8mm}{blue!20, fill=blue!20}
            {(v1.center) to (v3.center) to (v4.center) to (v2.center) to (v0.center) to (v1.center)}
        \end{scope}
        \begin{scope}
          \path[line cap=round, line join=round, line width=0.5pt, double=gray!5, double distance=5.5mm, draw=black]
            (v5.center) to (v7.center) to (v6.center);
          \highlight{3.8mm}{gray!5, fill=gray!5}
            {(v5.center) to (v7.center) to (v6.center) to (v7.center) to (v5.center)}
        \end{scope}
      \end{pgfonlayer}
    \end{tikzpicture}
    \caption{The sets $Z$ and $T$ as defined in Algorithm~\ref{alg:find_comp}.}
  \end{subfigure}
  \hfill
  \begin{subfigure}[t]{.25\textwidth}
    \centering
    \begin{tikzpicture}[scale=1.5]
      \node[draw, circle, fill=black, inner sep=2pt] (v0) at (0, 0) {};
      \node[draw, circle, fill=black, inner sep=2pt] (v1) at (0, 1) {};
      \node[draw, circle, fill=black, inner sep=2pt, label={[xshift=1mm]right:$v$}] (v2) at (1, 0) {};
      \node[draw, circle, fill=black, inner sep=2pt] (v3) at (0, 2) {};
      \node[draw, circle, fill=black, inner sep=2pt, label={[xshift=1mm]right:$u$}] (v4) at (1, 2) {};
      \node[draw, circle, fill=black, inner sep=2pt] (v5) at (-1, 1) {};
      \node[draw, circle, fill=black, inner sep=2pt] (v6) at (0, 3) {};
      \node[draw, circle, fill=black, inner sep=2pt] (v7) at (-1, 2) {};

      \draw[\niceArrow] (v2) to (v0);
      \draw[\niceArrow] (v2) to (v1);
      \draw[\niceArrow] (v4) to (v1);
      \draw[\niceArrow] (v4) to (v3);
      \draw[\niceArrow] (v1) to (v0);
      \draw[\niceArrow] (v1) to (v3);
      \draw[\niceArrow] (v1) to (v5);
      \draw[\niceArrow] (v3) to (v7);
      \draw[\niceArrow] (v3) to (v6);
      \draw[\niceArrow] (v7) to (v6);
      \draw[\niceArrow] (v4) to (v2);

      \begin{pgfonlayer}{background}
        \begin{scope}[opacity=.6, transparency group]
          \path[line cap=round, line join=round, line width=0.5pt, double=blue!20, double distance=3.8mm, draw=black]
            (v1.center) to (v3.center) to (v4.center) to (v2.center) to (v0.center) to (v1.center);
          \highlight{3.8mm}{blue!20, fill=blue!20}
            {(v1.center) to (v3.center) to (v4.center) to (v2.center) to (v0.center) to (v1.center)}
        \end{scope}
        \begin{scope}[opacity=.6, transparency group]
          \path[line cap=round, line join=round, line width=0.5pt, double=red!30, double distance=3.8mm, draw=black]
            (v3.center) to (v6.center) to (v7.center) to (v3.center);
          \highlight{3.8mm}{red!30, fill=red!30}
            {(v3.center) to (v6.center) to (v7.center)}
        \end{scope}
        \begin{scope}[opacity=.6, transparency group]
          \path[line cap=round, line join=round, line width=0.5pt, double=green!30, double distance=3.8mm, draw=black]
            (v1.center) to (v5.center);
          \highlight{3.8mm}{green!30, fill=green!30}
            {(v1.center) to (v5.center)}
        \end{scope}
      \end{pgfonlayer}
    \end{tikzpicture}
    \caption{$(2, 3)$-components after the insertion of the edge $uv$.}
  \end{subfigure}
  \caption{Impact of inserting edge $uv$ on the $(2, 3)$-components.}\label{fig:compMerge}
\end{figure}
\end{example}

The following lemma, a direct consequence of Corollary~\ref{corollary:comp_arises_unique}, will be used multiple times in the complexity analysis.
\begin{lemma}
\label{lemma:comp_tot_num}
During the execution of Algorithm~\ref{alg:comp}, only $O(n)$ $(k, \ell)$-components form.
\end{lemma}
\begin{proof}
The algorithm accepts $O(n)$ edges, and by Corollary~\ref{corollary:comp_arises_unique}, the insertion of each edge forms at most one new component.
\end{proof}

We conclude this section with a partial running-time analysis of Algorithm~\ref{alg:comp}.
\begin{proposition}
\label{proposition:alg_comp_complexity}
Ignoring the running times of the \textsc{InCommonComponent} and \textsc{UpdateComponents} subroutines, Algorithm~\ref{alg:comp} runs in $O(n^2 + m)$ time.
\end{proposition}
\begin{proof}
Consider the main \texttt{for} loop over the edges in Algorithm~\ref{alg:comp}.
For each rejected edge, the condition in line~\ref{alg:comp:check_in_comp} is not satisfied, so the body of the loop runs in $O(1)$ time.
Among all edges, only $O(n)$ are accepted, and for each such edge, the inner \texttt{while} loop takes at most $\ell + 1$ iterations, each of which runs in $O(n)$ time; additionally, the \textsc{FindComponent} subroutine is called once per accepted edge, also taking $O(n)$ time.
Thus, the total time for processing all accepted edges is $O(n^2)$, and the overall running time is $O(n^2 + m)$, as stated.
\end{proof}

\section{Implementation for the range $0 \leq \ell < 2k$}\label{sec:alg_general}

In this section, we describe the implementation of the \textsc{InCommonComponent} and \textsc{UpdateComponents} subroutines, thereby completing our $O(n^2 + m)$-time algorithm that solves the maximum-weight $(k, \ell)$-sparse subgraph problem for the full range of parameters $0 \leq \ell < 2k$.
We also provide a detailed analysis of the running time of the resulting algorithm.
\medskip

To efficiently track the $(k, \ell)$-components in $H = (V, F)$, we use a bit matrix $M \in \{0, 1\}^{V \times V}$, where $M_{u, v} = 1$ if and only if $u$ and $v$ lie in some common $(k, \ell)$-component.
Alongside $M$, we also maintain a list $\mathcal{C}$ containing the set of the current $(k, \ell)$-components, each represented as a list of vertices.
If $\ell < k$, we initialize $M$ to zero and set $\mathcal{C}$ to the empty list; otherwise, we set $M = I$ and initialize $\mathcal{C}$ with all singleton $(k, \ell)$-components, each consisting of a single vertex.
Throughout Algorithm~\ref{alg:comp}, we maintain $M$ and $\mathcal{C}$ so that the following invariants hold for the current subgraph $H$ at each iteration:
\begin{itemize}
\item For any $u, v \in V$, we have $M_{u, v} = 1$ if and only if there exists a $(k, \ell)$-component in $H$ containing both $u$ and $v$.
\item The list $\mathcal{C}$ contains exactly the $(k, \ell)$-components in $H$.
\end{itemize}

We now present the implementations of the \textsc{InCommonComponent} and \textsc{UpdateComponents} subroutines.

\begin{algorithm}[H]
\caption{In Common Component for $0 \leq \ell < 2k$}
\label{alg:in_common1}
  \begin{algorithmic}[1]
    \Input Two vertices $u, v \in V$.
    \Output \texttt{true} if there exists a $(k, \ell)$-component containing both $u$ and $v$; \texttt{false} otherwise.
    \vspace{0.2em}
    \hrule
    \vspace{0.2em}
    \Procedure{InCommonComponent$_{k, \ell}$}{$u, v$}
    \State \Return $M_{u, v} = 1$
    \EndProcedure
  \end{algorithmic}
\end{algorithm}
\vspace*{-4mm}

\begin{algorithm}[H]
\caption{Update Components for $0 \leq \ell < 2k$}
\label{alg:update1}
  \begin{algorithmic}[1]
    \Input A newly formed $(k, \ell)$-component $C$.
    \Effect Updates $M$ and $\mathcal{C}$ to restore the invariants above.
    \vspace{0.2em}
    \hrule
    \vspace{0.2em}
    \Procedure{UpdateComponents$_{k, \ell}$}{$C$}
    \State $U \gets \emptyset$ \Comment{Initialize the union of the components contained in $C$}
    \State $\mathcal{C}' \gets \{C\}$ \Comment{Initialize the new set of components}
    \For{$X \in \mathcal{C}$}
      \If{$X \subseteq C$} \Comment{Component contained in $C$?}
        \label{alg:update1:if_subset}
        \For{$(u, v) \in (U \setminus X) \times (X \setminus U)$} \Comment{Mark vertex pairs between $U$ and $X$}
          \label{alg:update1:for_descartes}
          \State $M_{u, v} \gets M_{v, u} \gets 1$
        \EndFor
        \State $U \gets U \cup X$ \Comment{Merge $X$ into $U$}
        \label{alg:update1:merge}
      \Else
        \State $\mathcal{C}' \gets \mathcal{C}' \cup \{X\}$ \Comment{Append $X$ to $\mathcal{C}'$}
        \label{alg:update1:append_comp}
      \EndIf
    \EndFor
    \For{$(u, v) \in U \times (C \setminus U)$} \Comment{Mark pairs with exactly one vertex in $U$}
      \label{alg:update1:for_descartes_2}
      \State $M_{u, v} \gets M_{v, u} \gets 1$
    \EndFor
    \For{$(u, v) \in (C \setminus U) \times (C \setminus U)$}  \Comment{Mark pairs completely outside $U$}
      \label{alg:update1:for_descartes_3}
      \State $M_{u, v} \gets 1$
    \EndFor
    \State $\mathcal{C} \gets \mathcal{C}'$ \Comment{Update $\mathcal{C}$ to contain the components of the new graph}
    \EndProcedure
  \end{algorithmic}
\end{algorithm}

\paragraph{Implementation details for \textnormal{\textsc{UpdateComponents}}.}
Both $U$ and $\mathcal{C}'$ are implemented as lists.
By Corollary~\ref{corollary:comp_intersection}, the condition $X \subseteq C$ in line~\ref{alg:update1:if_subset} can be tested by verifying whether $C$ contains the first two elements of $X$ (or, if $|X| = 1$, its sole vertex).
In line~\ref{alg:update1:for_descartes}, we generate the Cartesian product $(U \setminus X)\times(X \setminus U)$ by first computing the characteristic vectors $I_U$ and $I_X$ of $U$ and $X$, respectively.
Using these, the sets $U \setminus X$ and $X \setminus U$ can be extracted in $O(n)$ time so that two nested loops can enumerate all pairs in their Cartesian product.
The same method is applied to the loops in lines~\ref{alg:update1:for_descartes_2} and~\ref{alg:update1:for_descartes_3}.
The vector $I_U$ is reused when merging sets in line~\ref{alg:update1:merge} to ensure that each vertex $v$ in $X$ is appended to $U$ only once, that is, if $I_U(v)=0$.
$\bullet$
\medskip

By Corollary~\ref{corollary:comp_tot_size}, the list $\mathcal{C}$ requires $O(n)$ space throughout the execution of the algorithm.
The bit matrix $M$ contributes $O(n^2)$ additional space.
The graph traversal, Cartesian product generation, and characteristic vectors each require $O(n)$ space.
Thus, the total space complexity of the algorithm is $O(n^2)$, excluding the space required to store the input.

\medskip
The running-time analysis is more subtle.
The following theorem establishes our main result.
\begin{theorem}
\label{theorem:alg_comp_complexity}
Algorithm~\ref{alg:comp}, combined with the subroutines above, runs in $O(n^2 + m)$ time.
\end{theorem}
\begin{proof}
We analyze the time complexity of Algorithm~\ref{alg:comp} assuming the \textsc{FindComponent}, \textsc{InCommonComponent}, and \textsc{UpdateComponents} subroutines are implemented as described in Algorithms~\ref{alg:find_comp},~\ref{alg:in_common1},~and~\ref{alg:update1}, respectively.
By Proposition~\ref{proposition:alg_comp_complexity}, it suffices to bound the total running time of the \textsc{InCommonComponent} and \textsc{UpdateComponents} subroutines.
The \textsc{InCommonComponent} routine is called once per processed edge, runs in constant time, and hence contributes a total of $O(m)$ to the running time.
\medskip

It remains to analyze \textsc{UpdateComponents}.
Note that when the condition in line~\ref{alg:update1:if_subset} evaluates to true, the algorithm constructs characteristic vectors and extends the list $U$ in $O(n)$ time.
By Lemma~\ref{lemma:comp_tot_num}, this branch is entered only $O(n)$ times over the entire execution, contributing $O(n^2)$ time in total.
Observe that, by Corollary~\ref{corollary:comp_tot_size}, the time spent in line~\ref{alg:update1:append_comp} during a single call is $O(n)$.
Since \textsc{UpdateComponents} is called $O(n)$ times by Lemma~\ref{lemma:comp_tot_num}, these operations take $O(n^2)$ time in total.

The only remaining parts of \textsc{UpdateComponents} are the loops in lines~\ref{alg:update1:for_descartes},~\ref{alg:update1:for_descartes_2}, and~\ref{alg:update1:for_descartes_3}, which update the matrix $M$.
We must show that the total number of such updates is $O(n^2)$.
Let $\alpha$ denote the number of entries changed from $0$ to $1$, and $\beta$ the number of entries redundantly set from $1$ to $1$.
Since $M$ has $n^2$ entries, clearly $\alpha = O(n^2)$.
It remains to prove that $\beta$, the number of redundant writes, is also $O(n^2)$.
We establish this in two steps.

\begin{claim}
\label{claim:bound}
Let $C$ be the newly formed $(k, \ell)$-component passed to \textsc{UpdateComponents}, and let $C_1, \dots, C_t$ be the $(k, \ell)$-components that were fully contained in $C$ and thus deleted.
Let $U_0 = \emptyset$ and define $U_{i} = U_{i-1} \cup C_{i}$ for $i = 1, \dots, t$.
Then the number $\beta_C$ of redundant updates made during this subroutine call satisfies
\[
  \beta_C \leq n + \sum_{i=1}^t |U_{i-1} \cap C_{i}|^2.
\]
\end{claim}
\begin{proof}
Redundant updates arise only in the loops of lines~\ref{alg:update1:for_descartes} and~\ref{alg:update1:for_descartes_3}.
For the loop in line~\ref{alg:update1:for_descartes_3}, every redundant write occurs on the main diagonal of $M$, so their total number is bounded by~$n$.

Now consider the redundant writes made by the loop in line~\ref{alg:update1:for_descartes}.
Let $u,v \in C$ be distinct vertices such that the loop sets $M_{u,v}=1$ even though it was already $1$.
This situation can occur only if, before $C$ was formed, there existed a (unique) component $X$ containing both $u$ and $v$, i.e., $|X \cap C|\geq 2$.
By Lemma~\ref{lemma:intersection}, this implies that $X \cup C$ is a block.
Since $C$ is maximal, we must have $X \cup C = C$, hence $X \subseteq C$, and therefore $X$ coincides with one of the deleted components $C_p$.

Each redundant off-diagonal update contributed by the component $C_p$ corresponds to an ordered pair of distinct vertices $u, v \in C_p$ such that $M_{u,v}$ is redundantly set to $1$ during the current subroutine call.
The number of such pairs equals the $p$-th term in the first sum below, and summing over all $p$ yields
$$
  \sum_{i=1}^t 2 \cdot \binom{|U_{i-1} \cap C_{i}|}{2} \leq \sum_{i=1}^t |U_{i-1} \cap C_{i}|^2.
$$
Adding the diagonal contribution gives the stated upper bound on $\beta_C$.
\end{proof}

\begin{claim}
\label{claim:intersection}
With $C_1, \dots, C_t$ and $U_0, \dots, U_t$ defined as above,
\[
  \sum_{i=1}^t |U_{i-1} \cap C_{i}|^2 \leq 4t^2.
\]
\end{claim}
\begin{proof}
By the sparsity condition, $i_H(U_t) \leq k|U_t| - \ell$, and $i_H(C_{i}) = k|C_{i}| - \ell$ for each $i$ since each $C_{i}$ is a component.
By Corollary~\ref{corollary:comp_intersection}, each edge is induced by at most one component, so
$$
  \sum_{i=1}^t i_H(C_{i}) \leq i_H( C_1 \cup \dots \cup C_t ) = i_H(U_t).
$$
Combining these, we obtain
$$
  k \sum_{i=1}^t |C_{i}| - t \ell = \sum_{i=1}^t i_H(C_{i}) \leq i_H(U_t) \leq k|U_t| - \ell.
$$
Rearranging terms yields
$$
  \sum_{i=1}^t |U_{i-1} \cap C_{i}| = \sum_{i=1}^t \left( |C_{i}| - |C_{i} \setminus U_{i-1}| \right) = \sum_{i=1}^t |C_{i}| - |U_t| \leq \frac{\ell}{k} \cdot (t - 1) < 2(t - 1).
$$
Squaring both sides gives
$$
\sum_{i=1}^t |U_{i-1} \cap C_{i}|^2
  \leq \left( \sum_{i=1}^t |U_{i-1} \cap C_{i}| \right)^2
  \leq 4(t-1)^2 \leq 4t^2,
$$
which we had to show.
\end{proof}

Let $\mathcal{C}^*$ denote the set of $(k, \ell)$-components formed during the execution of the algorithm.
For each $C \in \mathcal{C}^*$, let $t_C$ be the number of previous $(k, \ell)$-components deleted when $C$ was formed.
By Claims~\ref{claim:bound} and~\ref{claim:intersection}, the number of redundant updates contributed by $C$ satisfies $\beta_C \leq 4t_C^2 + n$.
Lemma~\ref{lemma:comp_tot_num} implies that both $|\mathcal{C}^*|$ and the sum of all $t_C$ is $O(n)$, so by summing over all $(k, \ell)$-components, we obtain
$$
  \beta = \sum_{C \in \mathcal{C}^*} \beta_C \leq \sum_{C \in \mathcal{C}^*} (4t_C^2 + n) = \sum_{C \in \mathcal{C}^*} 4t_C^2 + n \cdot |\mathcal{C}^*| \leq 4 \left(\sum_{C \in \mathcal{C}^*} t_C\right)^2 + n \cdot |\mathcal{C}^*| = O(n^2).
$$
Thus $\beta = O(n^2)$, and together with $\alpha = O(n^2)$ we conclude that the matrix $M$ is updated $O(n^2)$ times overall.
This completes the proof of Theorem~\ref{theorem:alg_comp_complexity}, establishing that the algorithm indeed runs in $O(n^2 + m)$ time.
\end{proof}

\begin{note}
As noted earlier, the data structure originally proposed to accelerate the naive pebble game algorithm relies on a flawed running-time analysis.
In particular, it assumes that, after inserting an edge into the sought $(k, \ell)$-sparse subgraph, a generalization of the so-called \emph{bounded property} holds, that is,
\[
  |(C_1 \cup \dots \cup C_p) \cap (C_{p+1} \cup \dots \cup C_{p+q})| \leq 1,
\]
where $C_1, \dots, C_{p+q}$ denote arbitrary $(k, \ell)$-components to be deleted~\cite{pebble,pebbleDS}.
However, this assumption does not hold in general, as shown in the figure below.

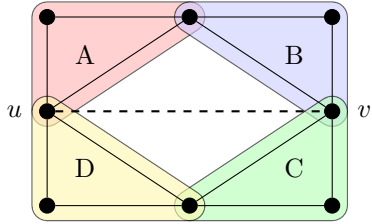
\begin{figure}[H]
  \centering
  \begin{tikzpicture}[scale=1.25]

    \node[draw, circle, fill=black, inner sep=2pt, label={[xshift=-1mm]left:$u$}] (A) at (0,1) {};
    \node[draw, circle, fill=black, inner sep=2pt] (B) at (1.5,2) {};
    \node[draw, circle, fill=black, inner sep=2pt, label={[xshift=1mm]right:$v$}] (C) at (3,1) {};
    \node[draw, circle, fill=black, inner sep=2pt] (D) at (1.5,0) {};

    \node[draw, circle, fill=black, inner sep=2pt] (E) at (0,2) {};
    \node[draw, circle, fill=black, inner sep=2pt] (F) at (3,2) {};
    \node[draw, circle, fill=black, inner sep=2pt] (G) at (3,0) {};
    \node[draw, circle, fill=black, inner sep=2pt] (H) at (0,0) {};

    \begin{pgfonlayer}{background}
    \begin{scope}[opacity=.6, transparency group]
      \path[line cap=round, line join=round, line width=0.5pt, double=red!30, double distance=3.8mm, draw=black]{(A.center) -- (B.center) -- (E.center) -- cycle};
      \highlight{3.8mm}{red!30, fill=red!30}{(A.center) to (B.center) to (E.center) -- cycle}
    \end{scope}
    \begin{scope}[opacity=.6, transparency group]
      \path[line cap=round, line join=round, line width=0.5pt, double=blue!20, double distance=3.8mm, draw=black]{(B.center) -- (C.center) -- (F.center) -- cycle};
      \highlight{3.8mm}{blue!20, fill=blue!20}{(B.center) to (C.center) to (F.center) -- cycle}
    \end{scope}
    \begin{scope}[opacity=.6, transparency group]
      \path[line cap=round, line join=round, line width=0.5pt, double=green!30, double distance=3.8mm, draw=black]{(C.center) -- (D.center) -- (G.center) -- cycle};
      \highlight{3.8mm}{green!30, fill=green!30}{(C.center) to (D.center) to (G.center) -- cycle}
    \end{scope}
    \begin{scope}[opacity=.6, transparency group]
      \path[line cap=round, line join=round, line width=0.5pt, double=yellow!40, double distance=3.8mm, draw=black]{(D.center) -- (A.center) -- (H.center) -- cycle};
      \highlight{3.8mm}{yellow!40, fill=yellow!40}{(D.center) to (A.center) to (H.center) -- cycle}
    \end{scope}
    \end{pgfonlayer}

    \draw (A) -- (B) -- (C) -- (D) -- (A);

    \draw (A) -- (E) -- (B);
    \draw (B) -- (F) -- (C);
    \draw (C) -- (G) -- (D);
    \draw (D) -- (H) -- (A);

    \draw[dashed, thick] (A) -- (C);

    \node at (0.4,1.6) {A};
    \node at (2.6,1.6) {B};
    \node at (2.6,0.4) {C};
    \node at (0.4,0.4) {D};
  \end{tikzpicture}
  \caption{Counterexample to the generalized bounded property.}\label{fig:counterexample_bounded}
\end{figure}

In Figure~\ref{fig:counterexample_bounded}, the $(2,3)$-components $A$, $B$, $C$, and $D$ are deleted after inserting edge $uv$, making the entire graph tight.
In this configuration, $|(A \cup B \cup C) \cap D| = 2 > 1$, violating the bounded property.
$\bullet$

\end{note}

\section{Optimized implementations for special cases}\label{sec:spec_cases}

For completeness, we describe optimized implementations for two important special cases of our problem.
Both approaches have been discussed in the literature~\cite{berg2003algorithms, pebble, pebbleDS}, but we present them here in a precise, self-contained form that integrates naturally into the framework of Algorithm~\ref{alg:comp}.
In each case, Algorithm~\ref{alg:comp} is extended with specialized implementations of the \textsc{InCommonComponent} and \textsc{UpdateComponents} subroutines.
These variants achieve the same $O(n^2 + m)$ running time as in the general case, but they require less space and are considerably easier to implement and analyze.

\subsection{The range $0 \leq \ell \leq k$}\label{sec:spec_cases:lower_range}

When $\ell \leq k$, the structure of $(k, \ell)$-sparse graphs simplifies significantly: their $(k, \ell)$-components cannot overlap.
The following lemma formalizes this observation and can be proved by a reasoning similar to that of Corollary~\ref{corollary:comp_intersection}.
\begin{lemma}
\label{lemma:intersection2}
If $\ell \leq k$, then the $(k, \ell)$-components in any $(k, \ell)$-sparse graph $H = (V, F)$ are pairwise disjoint.
\end{lemma}

By this lemma, each vertex is contained in at most one $(k, \ell)$-component.
Therefore, it suffices to store an array $R \in (V \cup \{\varnothing\})^V$, where $R(v)$ is a representative of the $(k, \ell)$-component containing $v$, or $\varnothing$ if $v$ does not belong to any $(k, \ell)$-component.
We initialize $R(v) = \varnothing$ for each $v \in V$ when $\ell < k$, and $R(v) = v$ when $\ell \geq k$.
Throughout Algorithm~\ref{alg:comp}, we maintain $R$ so that the following invariants hold for the current graph $H$ at each iteration:
\begin{itemize}
  \item For each $v \in V$, we have $R(v) \neq \varnothing$ if and only if there exists a $(k, \ell)$-component in $H$ containing $v$.
  \item For each $u, v \in V$ with $R(u) \neq \varnothing$ and $R(v) \neq \varnothing$, we have $R(u) = R(v)$ if and only if there exists a $(k, \ell)$-component in $H$ containing both $u$ and $v$.
\end{itemize}
Under these invariants, the implementation of the \textsc{InCommonComponent} and \textsc{UpdateComponents} subroutines simplify as follows.

\begin{algorithm}[H]
\caption{In Common Component for $0 \leq \ell \leq k$}\label{alg:in_common2}
  \begin{algorithmic}[1]
    \Input Two vertices $u, v \in V$.
    \Output \texttt{true} if there exists a $(k, \ell)$-component containing both $u$ and $v$; \texttt{false} otherwise.
    \vspace{0.2em}
    \hrule
    \vspace{0.2em}
    \Procedure{InCommonComponent$_{k, \ell}$}{$u, v$}
    \State \Return $R(u) \neq \varnothing$ \textbf{and} $R(v) \neq \varnothing$ \textbf{and} $R(u) = R(v)$
    \EndProcedure
  \end{algorithmic}
\end{algorithm}
\vspace*{-4mm}

\begin{algorithm}[H]
\caption{Update Components for $0 \leq \ell \leq k$}
\label{alg:update2}
  \begin{algorithmic}[1]
    \Input A newly formed $(k, \ell)$-component $C$.
    \Effect Updates $R$ so that all vertices of $C$ share the same representative.
    \vspace{0.2em}
    \hrule
    \vspace{0.2em}
    \Procedure{UpdateComponents$_{k, \ell}$}{$C$}
    \State $u \gets \operatorname{head}(C)$ \Comment{Choose the first vertex of $C$ as representative}
    \For{$v \in C$} \Comment{Set $u$ as the representative of each vertex in $C$}
    	\State $R(v) \gets u$
    \EndFor
    \EndProcedure
  \end{algorithmic}
\end{algorithm}

\begin{proposition}
Algorithm~\ref{alg:comp}, extended with the subroutines above, runs in $O(n^2 + m)$ time.
\end{proposition}
\begin{proof}
The \textsc{InCommonComponent} subroutine takes constant time, while \textsc{UpdateComponents} requires $O(n)$ time per call.
By Lemma~\ref{lemma:comp_tot_num}, the number of $(k, \ell)$-components formed during the execution of the algorithm is $O(n)$, so the total time spent in \textsc{UpdateComponents} is $O(n^2)$.
Combining this with Proposition~\ref{proposition:alg_comp_complexity}, we get an $O(n^2+m)$ running-time bound for the entire algorithm.
\end{proof}
Hence, this specialized algorithm preserves the $O(n^2 + m)$ running time of the general approach while reducing its space usage from $O(n^2)$ to $O(n)$, as the $n \times n$ matrix $M$ is no longer needed.
Note that this space bound excludes the space required to store the input graph.

\subsection{The unweighted case}\label{sec:spec_cases:unweighted}

For the weighted problem, Algorithm~\ref{alg:comp} must process the edges of the input graph in decreasing order of weight.
In the unweighted case, however, the edges can be processed in any order.
We exploit this flexibility by choosing, for each edge, an arbitrary endpoint, grouping the edges by this chosen endpoint, and then processing them vertex by vertex.
This enables a simpler and more space-efficient implementation.

Since we process the edges vertex by vertex, there is no need to maintain the full $n \times n$ bit matrix to determine whether an edge $uv$ can be inserted.
Instead, we maintain an array $I \in \{0, 1\}^V$ and the list $\mathcal{C}$ of the $(k, \ell)$-components, as in the general algorithm.
While processing the edges incident to a given vertex $u$, the value of $I(v)$ indicates whether there exists a $(k, \ell)$-component containing both $u$ and $v$.
If $\ell < k$, then we initialize $\mathcal{C}$ as an empty list; otherwise, we populate $\mathcal{C}$ with the singleton $(k, \ell)$-components.
When processing the edges incident to a vertex $u$, we maintain $I$ and $\mathcal{C}$ so that the following invariants hold for the current subgraph $H$ at each iteration:
\begin{itemize}
\item For each $v \in V$, we have $I(v) = 1$ if and only if there is a $(k, \ell)$-component in $H$ containing both $u$ and $v$.
\item The list $\mathcal{C}$ contains exactly the $(k, \ell)$-components in $H$.
\end{itemize}
The array $I$ is recomputed using a \textsc{Recalculate} subroutine whenever we move to a new vertex $u$.
A global variable $u_0 \in V\cup\{\varnothing\}$, initialized to $\varnothing$, records the first endpoint of the last processed edge to detect when this recomputation is needed.
With this setup in place, the implementations of the \textsc{Recalculate}, \textsc{InCommonComponent}, and \textsc{UpdateComponents} procedures are given below.

\begin{algorithm}[H]
\caption{Recalculate for the Unweighted Case}
\label{alg:recalculate}
  \begin{algorithmic}[1]
    \Input A vertex $u \in V$.
    \Effect Recomputes $I$ so that $I(v)=1$ if and only if there is a $(k, \ell)$-component containing both $u$ and $v$.
    \vspace{0.2em}
    \hrule
    \vspace{0.2em}
    \Procedure{Recalculate$_{k, \ell}$}{$u$}
    \State $I \gets 0$
    	\For{$X \in \mathcal{C}$}
    		\If{$u \in X$} \Comment{Component contains $u$?}
    			\For{$v \in X$} \Comment{Mark all vertices of $X$}
    				\State $I(v) \gets 1$
    			\EndFor
    		\EndIf
    	\EndFor
    \EndProcedure
  \end{algorithmic}
\end{algorithm}
\vspace{-3.5mm}

\begin{algorithm}[H]
\caption{In Common Component for the Unweighted Case}
\label{alg:in_common3}
  \begin{algorithmic}[1]
    \Input Two vertices $u, v \in V$.
    \Effect Update $u_0$ and recompute $I$ if $u$ changes.
    \Output \texttt{true} if there exists a $(k, \ell)$-component containing both $u$ and $v$; \texttt{false} otherwise.
    \vspace{0.2em}
    \hrule
    \vspace{0.2em}
    \Procedure{InCommonComponent$_{k, \ell}$}{$u, v$}
    \If{$u_0 \neq u$} \Comment{New first endpoint?}
    	\label{alg:in_common3:if_new_vertex}
    	\State $u_0 \gets u$
    	\State \Call{Recalculate$_{k, \ell}$}{$u$} \Comment{Recalculate $I$ with the new endpoint $u$}
    \EndIf
    \State \Return $I(v) = 1$
    \EndProcedure
  \end{algorithmic}
\end{algorithm}
\vspace{-3.5mm}

\begin{algorithm}[H]
\caption{Update Components for the Unweighted Case}
\label{alg:update3}
  \begin{algorithmic}[1]
    \Input A newly formed $(k, \ell)$-component $C$.
    \Effect Updates $I$ and $\mathcal{C}$ to restore the invariants above.
    \vspace{0.2em}
    \hrule
    \vspace{0.2em}
    \Procedure{UpdateComponents$_{k, \ell}$}{$C$}
    \State $\mathcal{C}' \gets \{C\}$ \Comment{Initialize the new set of components}
    \For{$X \in \mathcal{C}$}
    	\If{$X \not \subseteq C$} \Comment{Component not contained in $C$?}
    		\label{alg:update3:if_subset3}
        	\State $\mathcal{C}' \gets \mathcal{C}' \cup \{X\}$ \Comment{Append $X$ to $\mathcal{C}'$}
        	\label{alg:update3:append_comp3}
        \EndIf
    \EndFor
    \State $\mathcal{C} \gets \mathcal{C}'$ \Comment{Update $\mathcal{C}$ to contain the components of the new graph}
    \For{$v \in C$} \Comment{Mark vertices of the new component}
    	\State $I(v) \gets 1$
	\EndFor
    \EndProcedure
  \end{algorithmic}
\end{algorithm}

\paragraph{Implementation details for \textnormal{\textsc{UpdateComponents}}.}
The set $\mathcal{C}'$ is implemented as a list.
To test whether $C$ contains $X$ in line~\ref{alg:update3:if_subset3}, by Corollary~\ref{corollary:comp_intersection}, it suffices to check whether $C$ contains the first two vertices of $X$ (or, if $|X| = 1$, its sole vertex).
$\bullet$

\begin{proposition}
Algorithm~\ref{alg:comp}, extended with the subroutines above, runs in $O(n^2 + m)$ time.
\end{proposition}
\begin{proof}
Each call to \textsc{UpdateComponents} takes $O(n)$ time, and since the number of $(k, \ell)$-components is $O(n)$ in total, the combined time of all calls is $O(n^2)$.
By Corollary~\ref{corollary:comp_tot_size}, the running time of \textsc{Recalculate} is $O(n)$.
The branch in line~\ref{alg:in_common3:if_new_vertex} of \textsc{InCommonComponent} is entered at most $n$ times throughout the algorithm, resulting in an $O(n^2 + m)$ running time over all calls to the subroutine.
Combining this with Proposition~\ref{proposition:alg_comp_complexity}, we get an $O(n^2+m)$ running-time bound for the entire algorithm.
\end{proof}

As in the general case, the list $\mathcal{C}$ requires $O(n)$ space at any point during the algorithm.
The array~$I$ also takes linear space.
Thus, the space complexity of the algorithm --- excluding the space required to store the input --- is $O(n)$, thereby improving over the matrix-based method.
\vspace{-1mm}

\section*{Acknowledgment}
The authors are grateful to Tibor Jord\'an for his comments on a previous version of the manuscript.
This research has been implemented with the support provided by the Ministry of Innovation and Technology of Hungary from the National Research, Development and Innovation Fund, financed under the ELTE TKP 2021-NKTA-62 funding scheme, by the Ministry of Innovation and Technology NRDI Office within the framework of the Artificial Intelligence National Laboratory Program, by the Lend\"ulet Programme of the Hungarian Academy of Sciences --- grant number LP2021-1/2021, and by the Ministry of Innovation and Technology of Hungary from the National Research, Development and Innovation Fund --- grant number ADVANCED~150556.

\bibliographystyle{unsrt}
\bibliography{bibliography}

\end{document}